\documentclass[12pt,twoside]{article}
\usepackage[hscale=0.86,vscale=0.9,includehead,centering]{geometry}
\usepackage{clproject}
\setProjectProperty{paper.properties}

\newcommand{\theGenerationdDate}%
 {This file was generated on \today\ (by \texCompilerName{}
  with format \fmtname\ \fmtversion).}
\newcommand{\thePaperInfo}%
 {Version~\exactProjectVersion, no.~\releaseNumber.}

\usepackage[font=footnotesize,labelfont=bf]{caption}
\usepackage[pdfstartview=FitH,pagebackref]{hyperref}
\usepackage[noadjust]{cite}

\usepackage{graphicx}

\usepackage[cmex10]{amsmath}
\usepackage{amssymb,amsthm}
\usepackage{mathenv}

\newtheorem{theorem}{Theorem}[section]
\newtheorem{lemma}[theorem]{Lemma}
\newtheorem{proposition}[theorem]{Proposition}
\newtheorem{corollary}[theorem]{Corollary}
\newtheorem{definition}[theorem]{Definition}
\newtheorem{example}[theorem]{Example}

\newenvironment{proofof}[2][Proof of]%
 {\begin{proof}[#1 {#2}]}{\end{proof}}

\newcommand{\unitsof}[1]{{#1^{\times}}}
\newcommand{\field}[1]{\mathbb{F}_{#1}}
\newcommand{\nzfield}[1]{\unitsof{\field{#1}}}
\newcommand{\symgroup}[1]{\mathrm{S}_{#1}}
\newcommand{\mmgroup}{\mathfrak{M}}

\newcommand{\vt}[1]{\mathbf{#1}}
\newcommand{\mat}[1]{\mathbf{#1}}
\newcommand{\transpose}[1]{#1^{\mathsf{T}}}
\newcommand{\idm}{\mat{I}}

\DeclareMathOperator{\hd}{d_H}
\DeclareMathOperator{\wt}{wt}
\DeclareMathOperator{\wtd}{A}
\DeclareMathOperator{\wte}{W}
\DeclareMathOperator{\mmap}{m}
\DeclareMathOperator{\id}{id}
\DeclareMathOperator{\he}{H}
\newcommand{\floor}[1]{\left\lfloor{#1}\right\rfloor}
\newcommand{\ceil}[1]{\left\lceil{#1}\right\rceil}

\DeclareMathOperator{\ed}{d_E}
\DeclareMathOperator{\ew}{h}
\DeclareMathOperator{\ef}{h}

\newcommand{\eqdef}{:=}

\newcommand{\theTitle}{Entropy Distance}

\title{\theTitle%
 \thanks{\theGenerationdDate}
 \thanks{\thePaperInfo}
}
\author{Shengtian Yang%
 \thanks{S. Yang resides at Zhengyuan Xiaoqu 10-2-101,
  Hangzhou 310011, China (email: yangst@codlab.net).}
}
\date{}
 
\begin{document}

\maketitle

\begin{abstract}
Motivated by the approach of random linear codes,
a new distance in the vector space over a finite field is defined
as the logarithm of the ``surface area'' of a Hamming ball
with radius being the corresponding Hamming distance.
It is named entropy distance because of
its close relation with entropy function.
It is shown that entropy distance is a metric for a non-binary field
and a pseudometric for the binary field.
The entropy distance of a linear code is defined to be the smallest
entropy distance between distinct codewords of the code.
Analogues of the Gilbert bound, the Hamming bound,
and the Singleton bound are derived for the largest size
of a linear code given the length and entropy distance of the code.
Furthermore, as an important property related to lossless joint
source-channel coding, the entropy distance of a linear encoder
is defined.
Very tight upper and lower bounds are obtained for
the largest entropy distance of a linear encoder
with given dimensions of input and output vector spaces.
\end{abstract}

\begin{quote}
\end{quote}

\section{Introduction}

The aim of channel coding theory is to find effective ways of
combating noise so that information can be transmitted reliably
and quickly. One of the most important topics in this field
is about linear codes with large minimum distance,
because large minimum distance implies good error-correcting
capability (see e.g., \cite{Huffman200300, Lin200400}).

Let $\field{q}$ be a finite field of order $q = p^r$,
where $p$ is prime and $r \ge 1$.
The vector space of all $n$-tuples over $\field{q}$
is denoted by $\field{q}^n$. We usually write
a vector in $\field{q}^n$ in the row-vector form
$\vt{x} = (x_1, x_2, \ldots, x_n)$, and for $c \in \field{q}$
we denote by $\vt{c}$ the all-$c$ vector in $\field{q}^n$.
The \emph{(Hamming) distance} $\hd(\vt{x}, \vt{y})$ between
$\vt{x}, \vt{y} \in \field{q}^n$ is defined to be the number of
coordinates in which $\vt{x}$ and $\vt{y}$ differ.
In particular, we define the \emph{(Hamming) weight} $\wt(\vt{x})$ of
$\vt{x} \in \field{q}^n$ as $\hd(\vt{x}, \vt{0})$.
An $[n, k]$ \emph{linear code} $C$ over
$\field{q}$ is a $k$-dimensional subspace of $\field{q}^n$,
and a vector in $C$ is called a codeword of $C$.
The \emph{(minimum) distance} of $C$ is defined to be
the minimum of distances between distinct codewords of $C$,
or equivalently, the minimum weight of nonzero codewords of $C$.
Then an $[n, k]$ linear code with distance $d$ is usually
denoted as an $[n, k, d]$ linear code.

The significance of minimum distance is related with a classical
channel model called binary symmetric channel (BSC). Over a BSC,
the optimum decoding rule is to decode to the codeword closest
(in Hamming distance) to the received $n$-tuple, so a linear code
with distance $d$ can correct $(d-1)/2$ or fewer channel errors.
Note that the amount of information that
a linear code carries is characterized by its dimension $k$ or the
rate $k/n$, so one goal of coding theory is to determine the largest
rate of a linear code with a given distance (or the largest
distance of a linear code with a given rate).
There are countless papers on this topic (including nonlinear codes),
but so far, there is still a large gap between the best known
asymptotic lower bound and asymptotic upper bound on the rate of codes
(see e.g.,
\cite{Huffman200300, Gilbert195205, Varshamov195700, Tsfasman198200,
 McEliece197703, Jiang200408, Vu200509, Xing200101, Xing201109}
and the references therein).

This is a strange phenomenon, because on the channel coding problem,
information theory has provided very tight asymptotic lower and upper
bounds which in fact coincide at the point called channel capacity
(see e.g., \cite{Cover199100}). This implies that the coding problem
based on the distance of linear codes has diverged from its original
motivation for reliable transmission in the sense of information
theory. On the other hand, we note that the approach of random linear
codes (usually using a uniformly distributed random matrix) is
frequently used in theory to construct capacity-approaching coding
schemes or linear codes with large distance (see e.g.,
\cite{Elias195503, Gallager196300, Barg200209, Yang200904}
and the references therein).
If the approach of random linear codes is in the correct direction,
at least the author believes so, then do we need to rethink of
the distance of a linear code? Is it a good criterion of
error-correcting capability? Or can we learn something
valuable from the random-linear-code approach?

These questions motivate this paper, which will present a new
distance of (linear) code called entropy distance.
Roughly speaking, the entropy distance between
$\vt{x},\vt{y}\in\field{q}^n$ is defined as the logarithm of
the ``surface area'' of a sphere with radius $\hd(\vt{x},\vt{y})$,
and the entropy distance of a linear code is defined
in a similar way to Hamming distance.
A linear code with large entropy distance must have large (Hamming)
distance, but not vice versa.
Furthermore, we shall define the entropy distance of a linear encoder,
an interesting property related to lossless joint source-channel
coding.
Several lower and upper bounds about entropy distance of
linear codes and linear encoders are derived,
and concrete examples with large entropy distance are also provided.
In the case of linear encoders,
the lower and upper bounds on entropy distance turn out to be very tight.

The rest of this paper is organized as follows.
In Section~\ref{sec:M&D}, we revisit the sphere packing problem
in an information-theoretic manner
(by the approach of random linear codes).
In this process, we propose a sufficient condition
(called ``white'' condition) for universal packing.
To some extent, the minimum (Hamming) distance of a linear code
is a simplification of the ``white'' condition.
As another simplification, entropy distance is defined.
In Section~\ref{sec:LC.ED},
we investigate the properties of entropy distance of linear codes.
A lower bound and two upper bounds on the largest size of
a linear code with a given entropy distance are derived.
In Section~\ref{sec:LE.ED}, we goes further to define and study
the entropy distance of a linear encoder.
An upper bound and a lower bound on the largest entropy distance
of a linear encoder are derived
in terms of the dimensions of input and output vector spaces.
Concluding remarks are given in Section~\ref{sec:Conclusion}.

In the sequel, the multiplicative subgroup of nonzero elements of $\field{q}$ is denoted by $\nzfield{q}$.
The group of all permutations of the set $\{1, 2, \ldots, n\}$
is denoted by $\symgroup{n}$.
Each $\sigma \in \symgroup{n}$ together with each
$\vt{v}\in\nzfield{q}^n$ induces a \emph{monomial map}
$\mmap_{\sigma,\vt{v}}: \field{q}^n \to \field{q}^n$ given by
$\vt{x} \to (v_1x_{\sigma^{-1}(1)}, \ldots, v_nx_{\sigma^{-1}(n)})$.
In particular, $\mmap_{\sigma,\vt{1}}$ is called
\emph{coordinate permutation}
and is also denoted $\sigma$ for convenience.
The set of all monomial maps of $\field{q}^n$
is denoted by $\mmgroup(\field{q}^n)$.

For convenience of notation, we define
$aA\eqdef\{a\vt{x}: \vt{x}\in A\}$ and
$\vt{v}+A=A+\vt{v}\eqdef\{\vt{v}+\vt{x}: \vt{x}\in A\}$
for $a\in\field{q}$, $\vt{v}\in\field{q}^n$,
and $A\subseteq\field{q}^n$.

An $m$-by-$n$ matrix over a field is written as
$\mat{M}=(M_{i,j})_{m\times n}$
where $M_{i,j}$ denotes the $(i,j)$th entry.
The transpose of $\mat{M}$ is denoted by $\transpose{\mat{M}}$.
The $n\times n$ identity matrix is denoted $\idm_n$.

The \emph{identity function} on a set $A$ is denoted $\id_A:A\to A$
(given by $x\mapsto x$).
For a subset $B$ of $A$, the \emph{indicator function}
$1_B: A \to \{0,1\}$ is given by $x\mapsto 1$ for $x\in B$ and
$x\mapsto 0$ for $x\not\in B$. When the expression of $B$ is long,
we write $1B$ in place of $1_B(x)$.

For $x\in [0,1]$, we define the \emph{Hilbert entropy function} by
\[
\he_q(x) \eqdef x\log_q(q-1)-x\log_qx-(1-x)\log_q(1-x)
\]
with the convention $0\log_q0 = 0$. By $\he_q^{-1}$ we mean the inverse
of $\he_q$ from $[0, 1]$ to $[0, 1-q^{-1}]$.

The \emph{floor function} $\floor{x}$ and \emph{ceiling function}
$\ceil{x}$ of a real number $x$ are defined to be
the largest integer not greater than $x$
and the smallest integer not less than $x$, respectively.

Following the usual convention,
we always mean Hamming distance when we say distance, and entropy
distance should always be stated explicitly.

\section{Motivation and Definition}
\label{sec:M&D}

The essential of channel coding is related with a concept called
sphere packing. To some extent, it corresponds to the partition
induced by an optimal channel decoder.
Let $g: \field{q}^n \to \field{q}^k$ be a decoder,
and then the partition
$\{g^{-1}(\vt{x}): \vt{x} \in \field{q}^k\}$ of $\field{q}^n$
may be regarded as some kind of ``sphere'' packing. However,
the balls here are generally irregular and heterogeneous, or
should not be called ball at all.

In coding theory, we usually consider the packing problem of balls
in Hamming distance.
In the space $\field{q}^n$, a \emph{sphere} with
center $\vt{x} \in \field{q}^n$ and integer radius $r$ is the set
of all vectors which are all the same distance $r$ from $\vt{x}$.
The sphere together with its interior is called a \emph{ball}, i.e.,
the set $\{\vt{x}'\in\field{q}^n: \hd(\vt{x}', \vt{x}) \le r\}$.
The problem of finding the largest rate of codes with distance
$d$ is equivalent to the problem of finding the maximum number of
balls of radius $d/2$ that can be packed into the space $\field{q}^n$.
Obviously, this kind of balls is so regular that a large proportion
of the space is wasted in a general case.
It is by no means the kind of sphere packing
that information theory expects.

If we think in a manner more analogous to information
theory, for example, we may allow the ball contain some holes
as long as the total volume of holes is negligible in a certain sense,
then the situation changes drastically. By the approach of
random linear codes, we shall show that in this new sense,
there are linear codes whose sphere-packing radius is almost as
high as their distance, and that this kind of linear codes is
characterized by a weight distribution
that has almost the same shape as the function $\binom{n}{i}(q-1)^i$,
the ``surface area'' of a sphere with radius $i$.
For an $[n,k]$ linear code $C$ over $\field{q}$,
its weight distribution is a vector
$(\wtd_0(C), \ldots, \wtd_i(C), \ldots, \wtd_n(C))$
where $\wtd_i(C)$ is the number of codewords of weight $i$ in $C$.
The next two propositions conclude the existence of such a linear code.

\begin{proposition}[cf. \cite{Gallager196300, Barg200209}]
\label{pr:WDofRLC}
For $n \ge k \ge 1$, there is an $[n, k]$ linear code $C$ such that
\begin{equation}
\wtd_i(C) < n q^{-(n-k)} \binom{n}{i}(q-1)^i \qquad
\forall i = 1, 2, \ldots, n.
\label{eq:WDofRLC}
\end{equation}
\end{proposition}

\begin{proposition}
\label{pr:UPacking}
Let $C$ be an $[n,k]$ linear code satisfying \eqref{eq:WDofRLC}
and $S$ a subset of $\field{q}^n$ containing $\vt{0}$.
If $|S|< q^{n-k}/(2n)$, then there exist $f\in\mmgroup(\field{q}^n)$
and $B=B(f)\subseteq S$ such that
\begin{enumerate}%
\renewcommand{\labelenumi}{\textnormal{(\arabic{enumi})}}%
\renewcommand{\theenumi}{(\arabic{enumi})}
\item $\vt{0} \in B$,\label{enu:UPacking.1}
\item $|B| > |S|(1 - 2nq^{-(n-k)}|S|)$,\label{enu:UPacking.2}
\item The family $\{B_{\vt{c}}\}_{\vt{c} \in C}$ of sets
$B_{\vt{c}} \eqdef \vt{c} + f(B)$
is pairwise disjoint.\label{enu:UPacking.3}
\end{enumerate}
In particular, if $S$ is invariant under any monomial map and
$|S|< q^{n-k}/n$, then there exists $B\subseteq S$ such that
\begin{enumerate}%
\renewcommand{\labelenumi}{\textnormal{(\arabic{enumi})\ensuremath{'}}}
\renewcommand{\theenumi}{(\arabic{enumi})\ensuremath{'}}
\item $S_0 \eqdef
 \{\vt{s}\in S: \binom{n}{\wt(\vt{s})} (q-1)^{\wt(\vt{s})}
  \le q^{n-k}/(n|S|)\} \subseteq B$,\label{enu:UPacking.1s}
\item $|B| > |S|[1 - nq^{-(n-k)}(|S|-|S_0|)]$,\label{enu:UPacking.2s}
\item The family $\{B_{\vt{c}}\}_{\vt{c} \in C}$ of sets
$B_{\vt{c}} \eqdef \vt{c} + B$
is pairwise disjoint.\label{enu:UPacking.3s}
\end{enumerate}
\end{proposition}

Proposition~\ref{pr:UPacking} is more general than what we need,
so let us give some explanation.

It follows from Proposition~\ref{pr:WDofRLC} that
\[
\hd(C) \ge \min \left\{i:
\binom{n}{i} (q-1)^i \ge \frac{1}{n}q^{n-k}\right\}.
\]
Because $\binom{n}{i} (q-1)^i \le q^{n\he_q(i/n)}$
(see Lemma~\ref{le:Binom.Ineq}), we have
\begin{equation}
\frac{\hd(C)}{n} \ge \delta \eqdef
\he_q^{-1}\left(1 - \frac{k}{n} - \frac{\log_q n}{n}\right)
\in (0, 1-q^{-1}),\label{eq:delta}
\end{equation}
which is the well-known fact that random linear codes
achieve the asymptotic Gilbert-Varshamov (GV) bound
\cite{Huffman200300, Gilbert195205, Varshamov195700, Barg200209}.

Let $S$ be a ball in $\field{q}^n$ with
center $\vt{0}$ and radius $r = \floor{\delta n-\epsilon(n)}$.
The size of $S$ is
\[
\sum_{i=0}^r \binom{n}{i} (q-1)^i
< q^{n\he_q(r/n)} \le q^{n\he_q(\delta-\epsilon(n)/n)}
\]
by Lemma~\ref{le:Binom.Ineq}. Because $\he_q(x)$ is concave,
\begin{eqnarray*}[rcl]
\he_q(x) &\le &\he_q(x_0) + \he_q'(x_0) (x-x_0)\\
&= &\he_q(x_0) + (x-x_0) \log_q\frac{(q-1)(1-x_0)}{x_0},
\end{eqnarray*}
so that
\[
q^{n\he_q(\delta-\epsilon(n)/n)}
\le q^{n\he_q(\delta)+\epsilon(n)\log_q\gamma}
= n^{-1}q^{n-k} \gamma^{\epsilon(n)},
\]
where
\begin{equation}
\gamma \eqdef\frac{\delta}{(q-1)(1-\delta)} \in (0,1).\label{eq:gamma}
\end{equation}
This bound combined with Proposition~\ref{pr:UPacking} yields
the next corollary.

\begin{corollary}\label{co:Packing}
Let $C$ be an $[n,k]$ linear code satisfying \eqref{eq:WDofRLC}
and $S$ a ball in $\field{q}^n$ with
center $\vt{0}$ and radius
$r = \floor{\delta n-\epsilon(n)}$,
where $\delta$ is defined by \eqref{eq:delta} and $\epsilon(n)>0$.
Then there exists $B \subseteq S$ such that $\vt{0} \in B$,
$|B| > |S| (1 - \gamma^{\epsilon(n)})$,
and the family $\{\vt{c}+B\}_{\vt{c}\in C}$
of sets is pairwise disjoint,
where $\gamma$ is defined by \eqref{eq:gamma}.
(In particular, if we take $\epsilon(n) = \log_q n$, we obtain
a ``rough sphere'' packing with radius about $\delta n$,
almost as large as the distance of $C$.%
\footnote{``Rough sphere'' packing differs from list decoding
in that there is no uniform restriction on the number of codewords
within distance $r$ from every vector in $\field{q}^n$ although most
vectors have at most one codeword at distance $r$ or less from them.})
\end{corollary}

Note that the maximum possible size of a ``ball'' for packing is
$q^{n-k}$, so the above ``rough sphere'' packing is asymptotically
the best that we can do.
This implies that when we are seeking linear codes with large
distance, we should also be careful to check their packing radius
of ``rough sphere'', especially those codes exceeding the GV bound.
Recall that the packing radius of a so-called perfect code
is only about one half of its distance, and it cannot be improved by
``rough sphere'' packing.

It is natural to ask why a linear code satisfying
\eqref{eq:WDofRLC} has a large packing radius of ``rough sphere''.
Clearly, it is due to the shape of the weight distribution.
The weight distribution of a linear code is more
important than its distance.
However, we have little knowledge about other kinds of weight
distributions that also enable a linear code to have good capability
of ``rough sphere'' packing.%
\footnote{One of the candidates might be polar codes
\cite{Arikan200907} although their minimum distances are
asymptotically bad.}
But note that Proposition~\ref{pr:UPacking} indicates that
linear codes satisfying \eqref{eq:WDofRLC} have magic capability
of packing in a more general sense, that is, it allows the
shape of filler $S$ to be arbitrary.
This may be called universal packing, and in fact it is well known
that random linear codes are universal for channel coding,
an intrinsic property that can be found in almost every
information-theoretic proof based on random linear codes
(see e.g., \cite{Yang200904}).
Using a similar terminology in signal processing,
we call an $[n,k]$ linear code a ``white'' code
if its weight distribution is close to
the shape of $\binom{n}{i}(q-1)^i$, roughly in the form
\begin{equation}
\frac{\wtd_i(C)}{\binom{n}{i}(q-1)^i} \approx q^{-(n-k)}
\qquad \forall i=1,2,\ldots,n,\label{eq:White}
\end{equation}
a main part of \eqref{eq:WDofRLC}.
A ``white'' code is difficult to ``attack'', even by
a deliberately designed noise, because the codewords of such a code
is uniformly spread in the spectrum (i.e., weight distribution)
and hence can easily avoid the attack of noise
by randomly choosing a monomial map.

Now that a ``white'' code is so good,
why not using this criterion in code design?
However, determining the weight distribution of a linear code
is a very difficult task, which makes the criterion impractical.
The success of minimum distance of linear codes is partly because
it is easier to compute.
In fact, computing minimum distance of a linear code is equivalent to
determining the leftmost weight segment in which
the weight distribution is zero,
and we note that if the distance does not exceed the GV bound,
the weight distribution in that segment happens to be ``white''
by \eqref{eq:White}.
Then it is natural to ask if there is other ``zero'' weight segment
that also coincides with the ``white'' condition \eqref{eq:White}.
By checking \eqref{eq:White}, it is easy to find that
there is possibly a rightmost weight segment in which
the weight distribution is zero.
So as a compromise between minimum distance and \eqref{eq:White},
we may design a criterion that tracks the leftmost and rightmost
weight segments of zero weight distribution.
But in order to track these two segments,
we would need two parameters, say $(d_1, d_2)$, for example.
Can we find only one parameter to track both of these two segments?
Yes, we can. It is entropy distance.

\begin{definition}
\label{def:ed1}
The \emph{entropy distance} $\ed(\vt{x}, \vt{y})$ between
$\vt{x}, \vt{y} \in \field{q}^n$ is defined by
\[
\ed(\vt{x},\vt{y}) \eqdef \ef_{q,n}(\wt(\vt{x}-\vt{y})),
\]
where
\[
\ef_{q,n}(i) \eqdef \log_q \left[\binom{n}{i} (q-1)^{i}\right]
\qquad \text{for $i=0,1,\ldots,n$}.
\]
The \emph{entropy weight} $\ew(\vt{x})$ of
$\vt{x} \in \field{q}^n$ is defined as $\ed(\vt{x}, \vt{0})$.%
\footnote{A different definition of entropy weight is given
in \cite{Yang.M1} based on a variant of complete weight.
They are similar but motivated by different random coding techniques,
i.e., random monomial map and random coordinate permutation.
Obviously, the definition in this paper is better
because monomial maps include coordinate permutations
as a proper subset.}
The \emph{entropy distance} $\ed(C)$ of a
linear code $C$ is defined to be the smallest entropy distance
between distinct codewords of $C$, or equivalently, the minimum
entropy weight of nonzero codewords of $C$.
\end{definition}

At first glance, the definition of entropy distance may seem
very artificial, but the next propositions will convince the reader
that this definition is so natural that it qualifies as a metric
or pseudometric.
The name ``entropy distance'' comes from
the property~\ref{enu:h.entropy} in Proposition~\ref{pr:h.property}.

\begin{proposition}
\label{pr:h.property}
Let $q\ge 2$, $n \ge 1$, and $0\le i\le n$.
\begin{enumerate}%
\renewcommand{\labelenumi}{\textnormal{(\arabic{enumi})}}%
\renewcommand{\theenumi}{(\arabic{enumi})}
\item\label{enu:h.range} $0\le \ef_{q,n}(i) < n$.
\item\label{enu:h.max} Let $x_0\eqdef [(q-1)n-1]/q$ and then
\[
\ef_{q,n}(i)\lesseqqgtr \ef_{q,n}(i+1)
\qquad\text{for $i\lesseqqgtr x_0$}.
\]
The function $\ef_{q,n}(i)$ has one or two maxima
at $i=\ceil{x_0}, \floor{x_0}+1$.
\item \label{enu:h.entropy}
For $\alpha\in[0,1]$,
\[
\lim_{n\to\infty}\frac{1}{n}\ef_{q,n}(\floor{\alpha n})=\he_q(\alpha).
\]
\end{enumerate}
\end{proposition}

\begin{proposition}
\label{pr:ED.property}
Let $a \in \field{q}$ and $\vt{x},\vt{y},\vt{z} \in\field{q}^n$.
\begin{enumerate}%
\renewcommand{\labelenumi}{\textnormal{(\arabic{enumi})}}%
\renewcommand{\theenumi}{(\arabic{enumi})}
\item $0\le \ew(\vt{x}) < n$.
\item For $q=2$, $\ew(\vt{x}) = 0$
 if and only if $\vt{x}=\vt{0}$ or $\vt{1}$.\\
 For $q\ge 3$, $\ew(\vt{x}) = 0$ if and only if $\vt{x}=\vt{0}$.
\item $\ew(a\vt{x}) = |a|\ew(\vt{x})$ with $|a|\eqdef\wt(a)$.
\item \label{enu:ED.property.TI}
 $q^{\ew(\vt{x}+\vt{y})} \le \beta(\wt(\vt{x}),\wt(\vt{y}))
 q^{\ew(\vt{x})+\ew(\vt{y})}$, where
 \[
  \beta(w_1, w_2) \eqdef \frac{\max\{w_1, n-w_1, w_2, n-w_2\}}{n}
  \in [0.5,1].
 \]
\item $0\le \ed(\vt{x},\vt{y}) < n$. (Non-negativity)
\item For $q=2$, $\ed(\vt{x},\vt{y}) = 0$
 if and only if $\vt{x}=\vt{y}$ or $\vt{x}=\vt{y}+\vt{1}$.\\
 For $q\ge 3$, $\ed(\vt{x},\vt{y}) = 0$ if and only if $\vt{x}=\vt{y}$.
 (Identity of indiscernibles)
\item $\ed(\vt{x},\vt{y}) = \ed(\vt{y},\vt{x})$. (Symmetry)
\item $\ed(\vt{x},\vt{z}) \le \ed(\vt{x},\vt{y})+\ed(\vt{y},\vt{z})$.
  (Triangle inequality)\label{enu:ED.property.TI2}
\end{enumerate}
\end{proposition}

In the next few sections, we shall investigate the issue about
linear codes and linear encoders with large entropy distance
as an independent mathematical problem,
but the reader should keep in mind that entropy distance
is only a simplification of condition \eqref{eq:White}.

The proofs of results in this section are presented in
Appendix~\ref{app:M&D.proof}.

\section{Entropy Distance of Linear Codes}
\label{sec:LC.ED}

In this section, we shall investigate the properties of entropy
distance of linear codes, especially concerning those codes with
large entropy distance.
Let us begin with some examples about the entropy distance of
some familiar linear codes.

Recall that an $[n,k]$ linear code $C$ is characterized by
a $k\times n$ generator matrix $\mat{G}$
(such that $C=\{\vt{x}\mat{G}: \vt{x}\in\field{q}^k\}$)
or an $(n-k)\times n$ parity-check matrix $\mat{H}$ (such that
$C=\{\vt{x}\in\field{q}^n: \mat{H}\transpose{\vt{x}}=\vt{0}\}$).
A linear code is called the \emph{dual} of $C$ if its parity-check
(resp., generator) matrix is a generator (resp., parity-check)
matrix of $C$.
We denote the dual of $C$ by $C^{\perp}$.
The famous MacWilliams identities tell us that
\[
\wte_{C^{\perp}}(x,y) = \frac{1}{|C|} \wte_C(y-x, y+(q-1)x),
\]
where $\wte_C(x,y) \eqdef \sum_{i=0}^n \wtd_i(C)x^iy^{n-i}$
is the \emph{(homogeneous) weight enumerator} of $C$
(see e.g., \cite{Huffman200300}).

\begin{example}
Let $C$ be an $[n,1,n]$ repetition code whose generator matrix
is an $1\times n$ all-one matrix.
Then its weight enumerator is
\begin{equation}
\wte_C(x,y)=(q-1)x^n+y^n \label{eq:repcode.wte}
\end{equation}
and hence $\ed(C) = n\log_q(q-1)$.
\end{example}

\begin{example}
Let $C$ be an $[n,n-1,2]$ single parity-check code whose parity-check matrix
is an $1\times n$ all-one matrix, where $n \ge 2$.
By MacWilliams identities with \eqref{eq:repcode.wte}, we have
\[
\wte_C(x,y)=\frac{1}{q}\{(q-1)(y-x)^n+[y+(q-1)x]^n\},
\]
hence
\[
\wtd_1(C)=0,\quad
\wtd_2(C)=(q-1)\binom{n}{2},
\]
\[
\wtd_{n-1}(C)=\frac{n[(q-1)^{n-1}+(q-1)(-1)^{n-1}]}{q},\quad
\wtd_n(C)=\frac{(q-1)^n+(q-1)(-1)^{n}}{q},
\]
and therefore
\[
\ed(C)=\begin{cases}
\log_2 n &if $q=2$ and $n$ is odd,\\
0 &if $q=2$ and $n$ is even,\\
n\log_3 2 &if $q=3$ and $n=3,4,5$,\\
\ef_{q,n}(2) &otherwise.
\end{cases}
\]
\end{example}

\begin{example}
\label{ex:simpcode}
Let $C$ be a $[(q^k-1)/(q-1), k, q^{k-1}]$ simplex code whose
generator matrix consists of $(q^k-1)/(q-1)$ pairwise linearly
independent column vectors,
each chosen from a $1$-dimensional subspace of $\field{q}^k$.
By \cite[Theorem~2.7.5]{Huffman200300}, its weight enumerator is
\begin{equation}
\wte_C(x,y)=(q^k-1)x^{q^{k-1}}y^{(q^{k-1}-1)/(q-1)}+y^{(q^k-1)/(q-1)},
\label{eq:simpcode.wte}
\end{equation}
hence $\hd(C)=\ef_{q,(q^k-1)/(q-1)}(q^{k-1})$, which is the maximum of
$\ef_{q,(q^k-1)/(q-1)}$ by Proposition~\ref{pr:h.property}.
\end{example}

\begin{example}
Let $C$ be a $[(q^k-1)/(q-1), (q^k-1)/(q-1)-k,3]$ Hamming code,
the dual of a $[(q^k-1)/(q-1), k]$ simplex code, where $k\ge 2$.
Using MacWilliams identities with \eqref{eq:simpcode.wte}, we get
\[
\wte_C(x,y)=\frac{1}{q^k}\left\{
 (q^k-1)(y-x)^{q^{k-1}}[y+(q-1)x]^{(q^{k-1}-1)/(q-1)}
 +[y+(q-1)x]^{(q^{k}-1)/(q-1)}\right\},
\]
hence
\[
\wtd_1(C)=\wtd_2(C)=0,\quad
\wtd_3(C)>0,
\]
\[
\wtd_{(q^k-1)/(q-1)}(C)=\frac{(q-1)^{(q^{k-1}-1)/(q-1)}
 [(q-1)^{q^{k-1}} + (-1)^{q^{k-1}}(q^k-1)]}{q^k},
\]
and therefore
\[
\ed(C)=\begin{cases}
0 &if $q=2$,\\
5\log_4 3 &if $q=4$ and $k=2$,\\
\ef_{q,(q^k-1)/(q-1)}(3) &otherwise.
\end{cases}
\]
\end{example}

\begin{example}
Let $C$ be the $[2^m, \sum_{i=0}^r\binom{m}{i}, 2^{m-r}]$
$r$th order binary Reed-Muller (RM) code
(see e.g., \cite{Huffman200300, Lin200400}), where $0\le r\le m$.
Then $\ed(C)=0$ because the $0$th order binary RM code
which is contained in every $r$th order RM code
contains the all-one vector.
Certainly, it is easy to construct codes of large entropy distance
from RM codes with $r\ge 1$.
We may choose an arbitrary coordinate $i$ and let $C'$ be
the subcode of $C$ in which every codeword has symbol zero in
coordinate $i$.
Clearly, $C'$ is of dimension $\sum_{i=1}^r\binom{m}{i}$ and
its entropy distance is $\log_2 \binom{2^m}{2^{m-r}}$.
Puncturing $C'$ on coordinate $i$ further gives a code $C''$ of length
$2^m-1$ and entropy distance
\[
\min\left\{\log_2 \binom{2^m-1}{2^{m-r}},
\log_2 \binom{2^m-1}{2^m-2^{m-r}}\right\}
=\log_2 \binom{2^m-1}{2^{m-r}-1}.
\]
Indeed, the binary simplex code can be constructed in this way with
$r=1$.
\end{example}

Next, we present several bounds on the size of
a linear code with a given entropy distance.
For $0\le h\le\ef_{q,n}(\ceil{[(q-1)n-1]/q})$,
we denote by $D_q(n,h)$ the largest number of
codewords in a linear code over $\field{q}$ of length $n$
and entropy distance not less than $h$.
The next few propositions provide rather simple properties
of $D_q(n,h)$.

\begin{proposition}
\[
D_q(n,\ef_{n,q}(1))=\begin{cases}
2^{n-1} &if $q=2$,\\
q^n &otherwise.
\end{cases}
\]
\end{proposition}

The proof is left to the reader.

\begin{proposition}
\label{pr:Puncturing}
Let $[d_1, d_2]=\ef_{n,q}^{-1}([h,n))$.
If $d_1\ge 2$, then
\[
D_q(n,h)\le D_q(n-1,\min\{\ef_{q,n-1}(d_1-1),
 \ef_{q,n-1}(\min\{d_2, n-1\})\}).
\]
\end{proposition}

\begin{proof}
Let $C$ be a linear code of length $n$ and entropy distance at least
$h$ with $M$ codewords.
Then the range of distances between distinct codewords in $C$
are from $d_1$ to $d_2$.
Since $d_1\ge 2$, puncturing on any coordinate yields a code $C'$
also with $M$ codewords, and the distances between distinct codewords
in $C'$ are between $d_1-1$ and $\min\{d_2, n-1\}$,
so that the entropy distance of $C'$ is bounded below by
either $\ef_{q,n-1}(d_1-1)$ or $\ef_{q,n-1}(\min\{d_2, n-1\})$.
Therefore
\[
M\le D_q(n-1,\min\{\ef_{q,n-1}(d_1-1),
 \ef_{q,n-1}(\min\{d_2, n-1\})\}),
\]
and the proof is complete by letting $M=D_q(n,h)$.
\end{proof}

\begin{proposition}
\label{pr:Shortening}
Let $[d_1, d_2]=\ef_{n,q}^{-1}([h,n))$.
Then
\[
D_q(n,h)\le qD_q(n-1,\min\{\ef_{q,n-1}(d_1),
 \ef_{q,n-1}(\min\{d_2, n-1\})\}).
\]
\end{proposition}

\begin{proof}
Let $C$ be a linear code of length $n$ and entropy distance at least
$h$ with $M$ codewords.
Let $C(x)$ be the subcode of $C$ in which every codeword ends with
symbol $x$.
Then $C(0)$ contains at least $M/q$ codewords.
Puncturing this code on coordinate $n$ gives a code $C'$ of length
$n-1$, and the distances between distinct codewords in $C'$
are between $d_1$ and $\min\{d_2, n-1\}$, so that the entropy distance
of $C'$ is bounded below by either $\ef_{q,n-1}(d_1)$ or
$\ef_{q,n-1}(\min\{d_2, n-1\})$.
Therefore
\[
q^{-1}M\le D_q(n-1,\min\{\ef_{q,n-1}(d_1),
 \ef_{q,n-1}(\min\{d_2, n-1\})\}),
\]
and the proof is complete by letting $M=D_q(n,h)$.
\end{proof}

Now we shall derive several simple bounds on $D_q(n,h)$.
By convention, a lower bound $L(n,h)$ of $D_q(n,h)$ is said to
be achieved by some linear code $C$ if
$|C|\ge L(n,h)$ and $\ed(C)\ge h$.
If there is a family $\{C_i\}_{i=1}^\infty$ of linear codes
$C_i\subseteq \field{q}^{n_i}$
(supposing $n_i$ is strictly increasing in $i$) such that
\begin{subequations}
\label{eq:AsymptoticSense}
\begin{equation}
\liminf_{i\to\infty}
 \frac{\log_q|C_i|-\log_q L(n_i,n_i\bar{h})}{n_i}\ge 0
\end{equation}
and
\begin{equation}
\liminf_{i\to\infty}\frac{\ed(C_i)}{n_i}\ge \bar{h}
\end{equation}
\end{subequations}
for some $\bar{h}\in(0,1)$, then we say the lower bound
is asymptotically achieved by $\{C_i\}_{i=1}^\infty$.

The first is a lower bound, an analogue of the Gilbert bound
\cite{Huffman200300, Gilbert195205}.

\begin{theorem}
\label{th:Gilbert}
\begin{equation}
D_q(n,h) \ge \frac{q^n}{\sum_{i:\ef_{q,n}(i)<h} \binom{n}{i}(q-1)^i}.
\label{eq:Gilbert}
\end{equation}
\end{theorem}

\begin{proof}
Let $B\eqdef\{\vt{x}\in\field{q}^n: \ew(\vt{x})<h\}$.
It is clear that $aB\subseteq B$ for $a\in\field{q}$.
Let $C$ be a maximal linear code in the sense that
$\ed(C)\ge h$ and any larger linear code containing $C$
has entropy distance less than $h$.
Then By Lemma~\ref{le:Max.Sep},
$C$ is a maximal $B$-separable subspace of $\field{q}^n$,
and it satisfies $\bigcup_{\vt{c}\in C} (\vt{c}+B)=\field{q}^n$,
so that $|C| \ge q^n/|B|$,
which establishes the theorem.
\end{proof}

Just like the Gilbert bound, it is difficult to construct long codes
achieving \eqref{eq:Gilbert}.
But at least, we know that regular low-density parity-check codes
(with the row weight of parity-check matrix
being the logarithm of code length%
\footnote{In the binary case, the row weight must be odd.})
can achieve \eqref{eq:Gilbert} asymptotically,
an easy consequence of the analysis of weight distribution of
LDPC codes (see e.g., \cite[Theorem~5.6 and Remark~5.7]{Yang201111}).

The second is an upper bound, a simple modification of
the Hamming bound (see e.g., \cite[Theorem~1.12.1]{Huffman200300}).

\begin{theorem}
\label{th:Hamming}
\begin{equation}
D_q(n,h) \le \frac{q^n}{(1+1\{q=2\})\sum_{i=0}^t \binom{n}{i}(q-1)^i},
\label{eq:Hamming}
\end{equation}
where $t=\ceil{d/2}-1$ and $d$ is the smallest integer such that
$\ef_{n,q}(d)\ge h$.
\end{theorem}

\begin{proof}
Since the case of $q\ge 3$ is the same as
the original Hamming bound, we only prove the case of $q=2$.
Let $C$ be a (linear) code of length $n$ and entropy distance $h$.
Then by definition, the weight of all nonzero codewords
is between $d$ and $n-d$.
Let
\[
B_1\eqdef\{\vt{x}\in\field{q}^n: \wt(\vt{x})\le t\}
\]
and
\[
B_2\eqdef\{\vt{x}\in\field{q}^n: \wt(\vt{x})\ge n-t\}.
\]
Then for any $\vt{x}\in B_1$, $\vt{y}\in B_2$, and $\vt{c}\in C$,
we have
\[
\wt(\vt{x}-\vt{c}) \ge \wt(\vt{c})-\wt(\vt{x}) \ge d-t > t,
\]
\[
\wt(\vt{x}-\vt{c}) \le \wt(\vt{x})+\wt(\vt{c}) \le t+n-d < n-t,
\]
\[
\wt(\vt{y}-\vt{c}) \ge \wt(\vt{y})-\wt(\vt{c}) \ge n-t-(n-d) > t,
\]
\[
\wt(\vt{y}-\vt{c}) \le \wt(\vt{y}-\vt{1})+\wt(\vt{1}-\vt{c})
\le t+n-d < n-t.
\]
This implies that the family $\{\vt{c}+B\}_{\vt{c}\in C}$ of sets
with $B\eqdef B_1\cup B_2$ is pairwise disjoint, so that
$|B||C|\le q^n$, which establishes the theorem.
\end{proof}

The third is also an upper bound, an analogue of the Singleton bound
\cite{Singleton196404}.

\begin{theorem}
\label{th:Singleton}
Let $[d_1, d_2]=\ef_{n,q}^{-1}([h,n))$.
Then
\begin{equation}
D_q(n,h)\le q^{\min\{n-d_1+1,d_2\}}.\label{eq:Singleton}
\end{equation}
\end{theorem}

\begin{proof}
By the Singleton bound, it suffices to show that
$D_q(n,h)\le q^{d_2}$, which is obviously true by considering
the standard form
$\begin{pmatrix}\idm_k &\mat{A}\end{pmatrix}$
of generator matrix of an $[n,k]$ linear code.
\end{proof}

For illustration, we compute in Table~\ref{tab:LB.[7,k]}
the lower and upper bounds of the largest entropy distance of
a $[7,k]$ binary linear code for $1\le k\le 6$
as well as examples achieving the lower bound.
\begin{table*}[htbp]
\footnotesize
\begin{center}
\caption{The lower and upper bounds of the largest entropy distance of
 a $[7,k]$ binary linear code}\label{tab:LB.[7,k]}
\renewcommand{\arraystretch}{1.2}
\begin{tabular}{ccccc}
 \hline
 $k$ &The lower bound
  &The upper bound
  &Examples (generator matrix) &Entropy distance\\[-0.5ex]
 &(by \eqref{eq:Gilbert})
  &(by \eqref{eq:Hamming}, \eqref{eq:Singleton},
  and \eqref{eq:Special.1})\\
 \hline
 $1$ &$\log_2\binom{7}{3}$ &$\log_2\binom{7}{3}$
  &$\begin{pmatrix}1 &1 &1 &0 &0 &0 &0\end{pmatrix}$
  &$\log_2\binom{7}{3}$\\
 $2$ &$\log_2\binom{7}{2}$ &$\log_2\binom{7}{3}$
  &$\begin{pmatrix}
   1 &1 &1 &0 &0 &0 &0\\
   1 &0 &0 &1 &1 &0 &0
   \end{pmatrix}$
  &$\log_2\binom{7}{3}$\\[1.5ex]
 $3$ &$\log_2\binom{7}{2}$ &$\log_2\binom{7}{3}$
  &$\begin{pmatrix}
   1 &0 &1 &0 &1 &0 &1\\
   0 &1 &1 &0 &0 &1 &1\\
   0 &0 &0 &1 &1 &1 &1
   \end{pmatrix}$
  &$\log_2\binom{7}{3}$ (cf. Example~\ref{ex:simpcode})\\[2ex]
 $4$ &$\log_2\binom{7}{1}$ &$\log_2\binom{7}{2}$
  &$\begin{pmatrix}
   1 &0 &1 &0 &1 &0 &1\\
   0 &1 &1 &0 &0 &1 &1\\
   0 &0 &0 &1 &1 &1 &0\\
   0 &0 &1 &0 &0 &1 &0
   \end{pmatrix}$
  &$\log_2\binom{7}{2}$\\
 $5$ &$\log_2\binom{7}{1}$ &$\log_2\binom{7}{1}$
  &$\begin{pmatrix}
   \idm_{5} &\transpose{\vt{0}} &\transpose{\vt{0}}
   \end{pmatrix}$
  &$\log_2\binom{7}{1}$\\
 $6$ &$\log_2\binom{7}{1}$ &$\log_2\binom{7}{1}$
  &$\begin{pmatrix}
   \idm_{6} &\transpose{\vt{0}}
   \end{pmatrix}$
  &$\log_2\binom{7}{1}$\\
 \hline
\end{tabular}
\end{center}
\end{table*}

We close this section with a result on $D_2(n,\ef_{2,n}(2))$.

\begin{theorem}
For $n\ge 4$,
\begin{equation}
D_2(n,\ef_{2,n}(2)) =\begin{cases}
2^{n-3} &if $n$ is odd,\\
2^{n-2} &if $n$ is even.
\end{cases}
\label{eq:Special.1}
\end{equation}
\end{theorem}

\begin{proof}
The even case can be easily proved by
Theorem~\ref{th:Singleton} and the generator matrix
$\begin{pmatrix}
\idm_{n-2} &\transpose{\vt{1}} &\transpose{\vt{0}}
\end{pmatrix}$.

As for the odd case,
by Theorem~\ref{th:Singleton}, it suffices to show that
$D_2(n,\ef_{2,n}(2)) \ne 2^{n-2}$ and provide one example
of $[n,n-3]$ linear code of entropy distance at least $\ef_{2,n}(2)$.

We first show that $D_2(n,\ef_{2,n}(2)) \ne 2^{n-2}$.
If it were false, then there would exist an $[n,n-2]$ linear code $C$
of entropy distance at least $\ef_{2,n}(2)$.
Let
\[
\mat{G}\eqdef\begin{pmatrix}
\idm_{n-2} &\transpose{\vt{v}}_1 &\transpose{\vt{v}}_2
\end{pmatrix}
\]
be the standard form of the generator matrix with
$\vt{v}_1, \vt{v}_2\in \field{2}^{n-2}$.
Since the weight of the codeword $\vt{1}\mat{G}$ must not be
greater than $n-2$, we have $\vt{1}\mat{G}=(1,\ldots,1, 0, 0)$,
so that $\vt{v}_1$ and $\vt{v}_2$ must contain even number of ones.
Next, let $\{\vt{e}_k\}_{k=1}^{n-2}$ be the standard basis of
$\field{2}^{n-2}$.
Then the weight of the codeword
$\vt{e}_k\mat{G}=(\vt{e}_k, v_{1,k}, v_{2,k})$
must not be less than two,
so either $v_{1,k}$ or $v_{2,k}$ or both are one.
Because $n-2$ is odd, there would exist $k_0$ such that
$v_{1,k_0} = v_{2,k_0} = 1$.
However, the weight of $(\vt{1}-\vt{e}_{k_0}) \mat{G}$ would be
$(n-3)+2=n-1$, which is absurd.

For an example of $[n,n-3]$ linear code
of entropy distance at least $\ef_{2,n}(2)$,
consider the generator matrix
$\begin{pmatrix}
\idm_{n-3} &\transpose{\vt{1}} &\transpose{\vt{0}} &\transpose{\vt{0}}
\end{pmatrix}$,
which clearly has entropy distance $\ef_{2,n}(2)$.
\end{proof}

\section{Entropy Distance of Linear Encoders}
\label{sec:LE.ED}

In this section we shall define entropy distance in a more general
space, the direct product of two vector spaces.
In particular, we shall define and study the entropy distance of
a linear encoder.

Let $\field{q}^k\times\field{q}^n$ denote the direct product
of $\field{q}^k$ and $\field{q}^n$.
A vector $\vt{x}\in\field{q}^k\times\field{q}^n$
is written as
\[
(\vt{x}_1,\vt{x}_2)
\eqdef (x_{1,1}, \ldots, x_{1,k}, x_{2,1}, \ldots, x_{2,n}).
\]
For an all-$c$ vector in $\field{q}^k\times\field{q}^n$
with $c\in\field{q}$, we still write $\vt{c}$ for short.
A linear encoder $f:\field{q}^k\to\field{q}^n$
is a linear transformation from $\field{q}^k$ to $\field{q}^n$.
The rate of $f$ is defined to be $k/n$.
Usually $f$ is identified with its associated $k\times n$
transformation matrix, which is called generator matrix
in coding theory.
A linear encoder is said to be of full rank if its
generator matrix is of full rank.
A full-rank linear encoder is necessary for efficient information
processing because the full-rank condition ensures that
no information is lost during encoding (injective for $k\le n$)
or no vectors in the output vector space are wasted
(surjective for $k\ge n$).

\begin{definition}
\label{def:ed2}
The \emph{entropy distance} $\ed(\vt{x}, \vt{y})$ between
$\vt{x}, \vt{y} \in \field{q}^k\times\field{q}^n$
is defined by
\[
\ed(\vt{x},\vt{y}) \eqdef
\ed(\vt{x}_1,\vt{y}_1)+\ed(\vt{x}_2,\vt{y}_2).
\]
Likewise, the \emph{entropy weight} $\ew(\vt{x})$ of
$\vt{x}\in\field{q}^k\times\field{q}^n$
is defined by
\[
\ew(\vt{x}) \eqdef \ew(\vt{x}_1)+\ew(\vt{x}_2)
= \ed(\vt{x},\vt{0}).
\]
The \emph{entropy distance} $\ed(V)$ of a subspace $V$ of
$\field{q}^k\times\field{q}^n$ is defined to be the smallest
entropy distance between distinct vectors in $V$, or equivalently,
the minimum entropy weight of nonzero vectors in $V$.
Then the \emph{entropy distance} $\ed(f)$ of a linear encoder
$f:\field{q}^k\to\field{q}^n$ is defined to be the entropy distance
of its graph $\{(\vt{x}_1,f(\vt{x}_1)): \vt{x}_1\in\field{q}^k\}$.
\end{definition}

By Proposition~\ref{pr:ED.property},
it is easy to verify that the entropy distance in
$\field{q}^k\times\field{q}^n$ is a metric for $q\ge 3$
and a pseudometric for $q=2$.
The idea of Definition~\ref{def:ed2} comes from the author's
work on lossless joint source channel coding
\cite{Yang200904, Yang.M1}.
In a (distributed) lossless joint source-channel coding scheme based
on linear encoders, the sources are typically nonuniform (and
correlated), and hence the output of a linear encoder for small-weight
or small-entropy-weight input vectors is very important.
In other words, even for the same linear code, different generator
matrices may have very different performance.
It is found that a linear encoder that is (universally) good
(in the scheme proposed by \cite{Yang200904})
maps vectors of small entropy weight to
vectors of large entropy weight, an important property
now characterized by the entropy distance of a linear encoder.

Next, we study the lower and upper bounds on
the largest entropy distance of a full-rank linear encoder.
We denote by $E_q(k,n)$ the largest entropy distance of
a full-rank linear encoder $f:\field{q}^k\to\field{q}^n$.

Different from the entropy distance of a linear code, the entropy
distance of a linear encoder $f:\field{q}^k\to\field{q}^n$
has a very simple and tight upper bound:

\begin{theorem}
\label{th:encoder.ed.ub}
\begin{equation}
E_q(k,n) \le \begin{cases}
\ef_{2,n}(\ceil{\frac{n-1}{2}}) &if $q=2$,\\
\ef_{q,k}(1) + \ef_{q,n}\left(\ceil{\frac{(q-1)n-1}{q}}\right)
 &otherwise.
\end{cases}
\label{eq:encoder.ed.ub}
\end{equation}
\end{theorem}

The proof is left to the reader.

The (asymptotic) tightness of the upper bound is ensured by
the following lower bound.

\begin{theorem}
\label{th:encoder.ed.lb}
\begin{equation}
E_q(k,n) \ge h_0 \eqdef \max\left\{h:
 \sum_{\substack{i\ge 1,j\ge 1\{k\le n\}\\
  \ef_{q,k}(i)+\ef_{q,n}(j)<h}}
 \binom{k}{i}\binom{n}{j}(q-1)^{i+j} < (q-1)(q^n-q^{k'-1})\right\},
\label{eq:encoder.ed.lb}
\end{equation}
where $k'\eqdef\min\{k,n\}$.
\end{theorem}

\begin{proof}
Let $n'=\max\{k,n\}$, $l=n'-n$, and
\[
B\eqdef
\{(\vt{x},\vt{x}')\in (\field{q}^k\times\field{q}^n)\times\field{q}^l:
\ew(\vt{x})<h_0 \vee \vt{x}_1=\vt{0} \vee (\vt{x}_2,\vt{x}')=\vt{0}\}.
\]
It is clear that $aB\subseteq B$ for $a\in\field{q}$
and that
\begin{eqnarray*}[rcl]
|B|
&< &(q-1)(q^{n}-q^{k'-1})q^l+q^{n+l}+q^k-1\\
&= &q^{n'+1}+q^{k-1}-1.
\end{eqnarray*}
Then by Lemma~\ref{le:Max.Sep},
a maximal $B$-separable subspace $V$
of $\field{q}^k\times\field{q}^n\times\field{q}^l$ satisfies
\[
\bigcup_{\vt{v}\in V} (\vt{v}+B) =
\field{q}^k\times\field{q}^n\times\field{q}^l,
\]
so that
\begin{equation}
|V| \ge \frac{q^{k+n+l}}{|B|} > \frac{q^{k+n'}}{q^{n'+2}} = q^{k-2},
\label{eq:encoder.ed.lb.cover.1}
\end{equation}
that is, the dimension of $V$ is at least $k-1$.

For $(\vt{x}_1,\vt{x}_2,\vt{x}')\in
\field{q}^k\times\field{q}^n\times\field{q}^l$,
we define the canonical projections
\begin{eqnarray*}[rcl]
\pi_1(\vt{x}_1,\vt{x}_2,\vt{x}')&\eqdef&\vt{x}_1,\\
\pi_2(\vt{x}_1,\vt{x}_2,\vt{x}')&\eqdef&\vt{x}_2,\\
\pi_{23}(\vt{x}_1,\vt{x}_2,\vt{x}')&\eqdef&(\vt{x}_2,\vt{x}').
\end{eqnarray*}
Since the kernel of $\pi_1$ (resp., $\pi_{23}$) is a subset of $B$,
which intersects $V$ only at the zero vector,
the kernel of $\pi_1|_V$ (the restriction of $\pi_1$ to $V$)
(resp., $\pi_{23}|_V$) contains only the zero vector,
hence $\pi_1|_V$ (resp., $\pi_{23}|_V$) is injective,
and therefore the dimension of $V$ is at most $k$.

Let $S = (\pi_{1}|_V)(V)\times(\pi_{23}|_V)(V)$.
It is clear that $|S|\ge q^{2k-2}$ and that each
$(\vt{x}_1,\vt{x}_2,\vt{x}')\in S\setminus V$ is covered by
$(\pi_{1}|_V)^{-1}(\vt{x}_1) + B$ and
$(\pi_{23}|_V)^{-1}(\vt{x}_2,\vt{x}') + B$.
Then the bound \eqref{eq:encoder.ed.lb.cover.1} can be improved by
\[
|V| \ge \frac{q^{k+n+l}+|S|}{|B|+1}
> \frac{q^{k+n'}+q^{2k-2}}{q^{n'+1}+q^{k-1}} = q^{k-1},
\]
hence the dimension of $V$ is exactly $k$,
and therefore $\pi_1$ is an isomorphism.
If $k\ge n$, then $\pi_{23}$ is also an isomorphism.
Let $f$ be the composition $\pi_2 (\pi_1|_V)^{-1}$
from $\field{q}^k$ to $\field{q}^n$.
We conclude that $f$ is a full-rank linear encoder
of entropy distance not less than $h_0$.
The proof is complete.
\end{proof}

It is easy to see that the upper bound \eqref{eq:encoder.ed.ub}
is bounded above by $n+\log_q k+1$ and that the lower bound
\eqref{eq:encoder.ed.lb} is bounded below by
\[
\log_q\left(\frac{(q-1)^2q^{n-1}}{k(n+1)}\right)
> n-\log_q k-\log_q(n+1)-1.
\]
Then the gap between the two bounds is of order $\log_q (k^2n)$,
which is asymptotically negligible relative to $n$
if the rate $k/n$ is bounded.
Another fact to be noted is that the kernel and image of
a linear encoder achieving the lower bound \eqref{eq:encoder.ed.lb}
also achieve the lower bound \eqref{eq:Gilbert} asymptotically.

\begin{example}
Let $q=2$, $k=3$, and $n=7$.
By \eqref{eq:encoder.ed.ub} we have
\[
E_2(3,7) \le \ef_{2,7}(3) = \log_2 35.
\]
Since
\[
s\eqdef
\binom{3}{3}\binom{7}{7}
+\binom{3}{1}\binom{7}{7}+\binom{3}{2}\binom{7}{7}
+\binom{3}{3}\binom{7}{1}+\binom{3}{3}\binom{7}{6}
=21<124=2^7-2^2
\]
and
\[
s
+\binom{3}{3}\binom{7}{2}+\binom{3}{3}\binom{7}{5}
+\binom{3}{1}\binom{7}{1}+\binom{3}{1}\binom{7}{6}
+\binom{3}{2}\binom{7}{1}+\binom{3}{2}\binom{7}{6}
=147>124,
\]
it follows from \eqref{eq:encoder.ed.lb} that
\[
E_2(3,7) \ge \log_2 \binom{3}{3}\binom{7}{2} = \log_2 21.
\]
For an example achieving this lower bound,
consider the generator matrix of a $[7,3]$ simplex code
(cf. Example~\ref{ex:simpcode}).
Its entropy distance is $\log_2(\binom{3}{3}\binom{7}{4})=\log_2 35$.
\end{example}

Constructing linear encoders achieving the lower bound
\eqref{eq:encoder.ed.lb} is a difficult problem.
From the results in \cite{Yang.M1}, it follows that
an arbitrary linear encoder of a linear code with large entropy
distance concatenated with a low-density generator-matrix encoder
(with the column weight being the logarithm of dimension of output
vector space) can achieve \eqref{eq:encoder.ed.lb} asymptotically
(in a similar sense to \eqref{eq:AsymptoticSense}).

\section{Conclusion}
\label{sec:Conclusion}

In this paper, we proposed a new distance called entropy distance
for a linear code or a linear encoder.
The basic properties of entropy distance were investigated.
Several bounds on the entropy distance were derived.
In particular, we obtained the tight lower and upper bounds on
the largest entropy distance of a full-rank linear encoder
(Theorems~\ref{th:encoder.ed.ub} and \ref{th:encoder.ed.lb}).
Some concrete examples of linear codes and encoders
with large entropy distance were also provided.

As a mathematical problem, entropy distance brings many interesting
issues, some of which are not easier than their counterparts in
Hamming distance, e.g., determining the tight lower and upper bounds
on the largest size of a linear code given the length and entropy
distance of the code
(cf. Theorems~\ref{th:Gilbert} and \ref{th:Hamming}).
On the other hand, the significance of entropy distance for coding
applications, which remains for future study,
is still far from being understood.

\appendix

\section{Lemmas}

\begin{lemma}[cf. {\cite[p.~284]{Cover199100}}]
\label{le:Binom.Ineq}
Let $q\ge 2$, $n\ge 1$, and $0\le k\le n$.
Then
$\sum_{i \in I} \binom{n}{i}(q-1)^i \le q^{n\he_q(k/n)}$
with equality if and only if $k=(q-1)n/q$, where
\[
I\eqdef \left\{0\le i\le n:
\left[\frac{k}{(q-1)(n-k)}\right]^{k-i}\le 1\right\}
= \begin{cases}
\{0, 1, \ldots, k\} &if $k < (q-1)n/q$,\\
\{0, 1, \ldots, n\} &if $k = (q-1)n/q$,\\
\{k, k+1, \ldots, n\} &if $k > (q-1)n/q$.
\end{cases}
\]
\end{lemma}

\begin{proof}
Using the binomial formula
$\sum_{i=0}^n \binom{n}{i} x^i (1-x)^{n-i} = 1$
with $x=k/n$, we get
\begin{eqnarray*}[rcl]
1 &\ge &\sum_{i\in I} \binom{n}{i}(q-1)^i
\left[\frac{k}{(q-1)n}\right]^i \left(1-\frac{k}{n}\right)^{n-i}\\
&\ge &\sum_{i\in I} \binom{n}{i}(q-1)^i
\left[\frac{k}{(q-1)n}\right]^k \left(1-\frac{k}{n}\right)^{n-k}\\
&\ge &q^{-n\he_q(k/n)} \sum_{i\in I} \binom{n}{i}(q-1)^i
\end{eqnarray*}
and therefore $\sum_{i \in I} \binom{n}{i}(q-1)^i \le q^{n\he_q(k/n)}$
with equality if and only if $I=\{0,1,\ldots,n\}$.
\end{proof}

\begin{lemma}
\label{le:Max.Sep}
Let $B$ be a subset of $\field{q}^n$ such that
$aB\subseteq B$ for $a\in\field{q}$.
A subset $S$ of $\field{q}^n$ is said to be $B$-separable if
$S \cap (\vt{s}+B) = \vt{s}$ for each $\vt{s}\in S$.
A $B$-separable subspace $V$ of $\field{q}^n$ is said to be maximal if
any larger subspace containing $V$ is not $B$-separable.
Then a maximal $B$-separable subspace $V$
satisfies $\bigcup_{\vt{v}\in V} (\vt{v}+B) = \field{q}^n$.
\end{lemma}

\begin{proof}
First note that for a vector space $V$, the $B$-separable condition
is reduced to $V\cap B=\{\vt{0}\}$.
We suppose $V\ne\field{q}^n$ and choose any $\vt{x}\notin V$.
Since $V$ is maximal, the subspace
$V'\eqdef\{a\vt{x}+\vt{v}: a\in\field{q}, \vt{v}\in V\}$
is not $B$-separable, so that $V'\cap B$ contains a nonzero vector
$\vt{x}'=a\vt{x}+\vt{v}'$ for some $a\in\field{q}\setminus\{0\}$
and $\vt{v}'\in V$, and hence $\vt{x}=a^{-1}\vt{x}'-a^{-1}\vt{v}'$.
The proof is complete by noting that
$-a^{-1}\vt{v}'\in V$ and $a^{-1}\vt{x}'\in B$.
\end{proof}

\section{The Proofs of Results in Section~\ref{sec:M&D}}
\label{app:M&D.proof}

\begin{proofof}{Proposition~\ref{pr:WDofRLC}}
For any $v \in \nzfield{q^n}$, we define the linear transformation
$f_v: \field{q^n} \to \field{q^n}$ given by $x \mapsto v x$, which
is also a linear transformation of $\field{q}^n$ onto $\field{q}^n$.
Let $g$ be an arbitrary injective linear transformation from
$\field{q}^k$ to $\field{q}^n$.
An $[n,k]$ linear code $C_v$ is defined to be the image
$f_v(g(\field{q}^k))$. Let us compute the average weight
distribution of $C_v$ over all $v \in \nzfield{q^n}$ for nonzero weight.
\begin{eqnarray}[rcl]
\frac{1}{|\nzfield{q^n}|} \sum_{v\in\nzfield{q^n}} \wtd_i(C_v)
&= &\frac{1}{|\nzfield{q^n}|} \sum_{v\in\nzfield{q^n}}
\sum_{\vt{y}\in\field{q}^n:\wt(\vt{y})=i}
\sum_{\vt{x}\in\field{q}^k\setminus\{\vt{0}\}}
1\{f_v(g(\vt{x})) = \vt{y}\}\nonumber\\
&= &\sum_{\vt{x}\in\field{q}^k\setminus\{\vt{0}\}}
\sum_{\vt{y}\in\field{q}^n:\wt(\vt{y})=i}
\frac{1}{|\nzfield{q^n}|} \sum_{v\in\nzfield{q^n}}
1\{vg(\vt{x}) = \vt{y}\}\nonumber\\
&= &(q^{n}-1)^{-1}(q^k-1)\binom{n}{i}(q-1)^i.\label{eq:WDofRLC.p1}
\end{eqnarray}
It is then easy to show that there is a linear code $C_v$ such that
\eqref{eq:WDofRLC} holds. If it were false, then for every
$v\in\nzfield{q^n}$ there would exist $i\ne 0$ such that
\[
\wtd_i(C_v) \ge n q^{-(n-k)} \binom{n}{i}(q-1)^i,
\]
so that there exists at least one $j$ such that more than $(q^n-1)/n$
linear codes of $\{C_v: v\in\nzfield{q^n}\}$ satisfy
\[
\wtd_j(C_v) \ge n q^{-(n-k)} \binom{n}{j}(q-1)^j,
\]
and therefore the average weight distribution of
$\{C_v: v\in\nzfield{q^n}\}$ for weight $j$ should be no less than
\[
q^{-(n-k)} \binom{n}{j}(q-1)^j,
\]
which is absurd by \eqref{eq:WDofRLC.p1}. The proof is complete.
\end{proofof}

\begin{proofof}{Proposition~\ref{pr:UPacking}}
For $f\in\mmgroup(\field{q}^n)$, we define a family of sets
$S_{f,\vt{c}} \eqdef \vt{c} + f(S)$ for $\vt{c} \in C$.
Because these sets are homogeneous, it suffices to focus on one set,
for example, $S_{f,\vt{0}}$.
We define the function $\Phi_f: S \to \{0, 1\}$ by
\[
\Phi_f(\vt{s}) \eqdef
1\{\text{there exists $\vt{c}\in C\setminus\{\vt{0}\}$
such that $f(\vt{s})=\vt{c}+f(\vt{s}')$
for some $\vt{s}'\in S$}\}.
\]
Then the number of elements in
$S_{f,\vt{0}}$ that are overlapped with another $S_{f,\vt{c}}$
for some $\vt{c}\ne\vt{0}$ is $\sum_{\vt{s} \in S} \Phi_f(\vt{s})$.

Note that $\Phi_f(\vt{s})$ can be bounded above by
\[
U_f(\vt{s}) \eqdef \sum_{\vt{c}\in C\setminus\{\vt{0}\}}
\sum_{\vt{s}'\in S} 1\{f(\vt{s})=\vt{c}+f(\vt{s}')\}.
\]
Then the average $\Phi(\vt{s})$ of $\Phi_f(\vt{s})$ over
all $f\in\mmgroup(\field{q}^n)$ is bounded by
\[
\frac{1}{|\mmgroup(\field{q}^n)|} \sum_{f\in\mmgroup(\field{q}^n)}
U_f(\vt{s})
= \sum_{\vt{s}'\in S} U_{\vt{s}, \vt{s}'},
\]
where
\[
U_{\vt{s}, \vt{s}'} \eqdef \frac{1}{|\mmgroup(\field{q}^n)|}
\sum_{\vt{c}\in C\setminus\{\vt{0}\}}
\sum_{f\in\mmgroup(\field{q}^n)} 1\{f(\vt{s}-\vt{s}')=\vt{c}\}.
\]
It is easy to show that
\[
U_{\vt{s}, \vt{s}'} = \begin{cases}
0 &if $\vt{s}=\vt{s}'$,\\
\displaystyle\frac{\wtd_{\wt(\vt{s}-\vt{s}')}(C)}%
 {\binom{n}{\wt(\vt{s}-\vt{s}')} (q-1)^{\wt(\vt{s}-\vt{s}')}}
 &otherwise.
\end{cases}
\]
From \eqref{eq:WDofRLC} it follows that
$U_{\vt{s}, \vt{s}'} < nq^{-(n-k)}$ for $\vt{s}\ne\vt{s}'$,
so
\begin{equation}
\Phi(\vt{s}) < nq^{-(n-k)}(|S|-1).\label{eq:UPacking.p1}
\end{equation}

If $|S|<q^{n-k}/(2n)$, we have
\[
\Phi(\vt{0}) < \frac{1}{2}
\]
and
\[
\sum_{\vt{s}\in S\setminus\{\vt{0}\}} \Phi(\vt{s})
< nq^{-(n-k)}(|S|-1)^2.
\]
By a similar argument to Proposition~\ref{pr:WDofRLC}, we conclude
that there exists $g\in\mmgroup(\field{q}^n)$ such that
\[
\Phi_g(\vt{0}) = 0
\]
and
\[
\sum_{\vt{s}\in S\setminus\{\vt{0}\}} \Phi_g(\vt{s})
< 2nq^{-(n-k)}(|S|-1)^2.
\]
If we choose $B = \Phi_g^{-1}(0)$, then it is clear that
conditions~\ref{enu:UPacking.1}--\ref{enu:UPacking.3} hold.

If $S$ is invariant under any monomial map, then
\begin{eqnarray*}[rcl]
\Phi(\vt{s})
&= &\frac{1}{|\mmgroup(\field{q}^n)|} \sum_{f\in\mmgroup(\field{q}^n)}
 \Phi_{\id_{\field{q}^n}}(f(\vt{s}))\\
&= &\frac{\sum_{\vt{s}':\wt(\vt{s}')=\wt(\vt{s})}
 \Phi_{\id_{\field{q}^n}}(\vt{s}')}%
 {\binom{n}{\wt(\vt{s})} (q-1)^{\wt(\vt{s})}}
\end{eqnarray*}
and hence, conditions~\ref{enu:UPacking.1s}--\ref{enu:UPacking.3s}
follow from \eqref{eq:UPacking.p1} with
$B=\Phi_{\id_{\field{q}^n}}^{-1}(0)$. The proof is complete.
\end{proofof}

\begin{proofof}{Proposition~\ref{pr:h.property}}
\ref{enu:h.range} The inequality can be rewritten as
\[
1 \le \binom{n}{i}(q-1)^i < q^n.
\]
The first inequality is clearly true, and the second comes from
$q^n = \sum_{i=0}^n \binom{n}{i}(q-1)^i$.

\ref{enu:h.max} The statement is proved by observing that
\[
\ef_{q,n}(i+1) - \ef_{q,n}(i) = \Delta(i)
\eqdef \log_q \frac{(q-1)(n-i)}{(i+1)}
\]
and
\[
\Delta(i)\gtreqqless 0
\qquad\text{for $i\lesseqqgtr \frac{(q-1)n-1}{q}$}.
\]

\ref{enu:h.entropy} See Lemma~\ref{le:Binom.Ineq}
 and \cite[Theorem~12.1.3]{Cover199100}.
\end{proofof}

\begin{proofof}{Proposition~\ref{pr:ED.property}}
We only prove \ref{enu:ED.property.TI}.
The proofs of other statements are left to the reader.

\ref{enu:ED.property.TI} We first prove the inequality in the case of
$\wt(\vt{x}+\vt{y})=0$ or $n$.
If $\wt(\vt{x}+\vt{y})=0$, then $\vt{x}=-\vt{y}$,
so that
\[
q^{\ew(\vt{x}+\vt{y})} = 1
\le \frac{\max\{\wt(\vt{x}), n-\wt(\vt{x})\}}{n} q^{\ew(\vt{x})}
\le \beta(\wt(\vt{x}),\wt(\vt{y})) q^{\ew(\vt{x})+\ew(\vt{y})}.
\]
If $\wt(\vt{x}+\vt{y})=n$, then $\wt(\vt{x})+\wt(\vt{y})\ge n$, so that
\[
q^{\ew(\vt{x}+\vt{y})} = (q-1)^n
\le \beta(\wt(\vt{x}),\wt(\vt{y})) q^{\ew(\vt{x})+\ew(\vt{y})}.
\]

Now we shall prove the inequality by induction on $n$.
The case of $n=1$ has already been covered by the above special cases.
If $n\ge 2$, we can assume that
$1\le \wt(\vt{x}+\vt{y})\le n-1$.
With no loss of generality, we assume that
$\wt(x_1+y_1)=0$ and $\wt(x_n+y_n)=1$,
and we define
\begin{eqnarray*}[rclqrcl]
x' &= &(x_2, \ldots, x_n),&
x''&= &(x_1, \ldots, x_{n-1}),\\
y' &= &(y_2, \ldots, y_n),&
y''&= &(y_1, \ldots, y_{n-1}).
\end{eqnarray*}
Supposing the inequality is true for $n-1$, we have
\begin{eqnarray*}[rcl]
q^{\ew(\vt{x}+\vt{y})}
&= &\left(\binom{n-1}{\wt(\vt{x}+\vt{y})}
 + \binom{n-1}{\wt(\vt{x}+\vt{y})-1}\right)
 (q-1)^{\wt(\vt{x}+\vt{y})}\\
&= &q^{\ew(\vt{x'}+\vt{y'})} + (q-1)q^{\ew(\vt{x''}+\vt{y''})}\\
&\le &q^{\ew(\vt{x'})+\ew(\vt{y'})}
 + (q-1)q^{\ew(\vt{x''})+\ew(\vt{y''})}\\
&= &\binom{n-1}{\wt(\vt{x})-\wt(x_1)}
 \binom{n-1}{\wt(\vt{y})-\wt(y_1)}
 (q-1)^{\wt(\vt{x})+\wt(\vt{y})-\wt(x_1)-\wt(y_1)}\\
& &+ (q-1)\binom{n-1}{\wt(\vt{x})-\wt(x_n)}
 \binom{n-1}{\wt(\vt{y})-\wt(y_n)}
 (q-1)^{\wt(\vt{x})+\wt(\vt{y})-\wt(x_n)-\wt(y_n)}\\
&\le &\max\left\{\frac{(n-\wt(\vt{x}))(n-\wt(\vt{y}))}{n^2},
 \frac{\wt(\vt{x})\wt(\vt{y})}{n^2(q-1)^2}\right\}
 q^{\ew(\vt{x})+\ew(\vt{y})}\\
& &+ (q-1) \max\left\{\frac{(n-\wt(\vt{x}))\wt(\vt{y})}{n^2(q-1)},
 \frac{\wt(\vt{x})(n-\wt(\vt{y}))}{n^2(q-1)},
 \frac{\wt(\vt{x})\wt(\vt{y})}{n^2(q-1)^2}1\{q\ge 3\}\right\}
 q^{\ew(\vt{x})+\ew(\vt{y})}\\
&\le &\beta(\wt{\vt(x)},\wt{\vt(y)}) q^{\ew(\vt{x})+\ew(\vt{y})},
\end{eqnarray*}
as desired.
\end{proofof}

\bibliographystyle{IEEEtran}
\bibliography{IEEEabrv,mined}

\begin{thebibliography}{10}
\providecommand{\url}[1]{#1}
\csname url@samestyle\endcsname
\providecommand{\newblock}{\relax}
\providecommand{\bibinfo}[2]{#2}
\providecommand{\BIBentrySTDinterwordspacing}{\spaceskip=0pt\relax}
\providecommand{\BIBentryALTinterwordstretchfactor}{4}
\providecommand{\BIBentryALTinterwordspacing}{\spaceskip=\fontdimen2\font plus
\BIBentryALTinterwordstretchfactor\fontdimen3\font minus
  \fontdimen4\font\relax}
\providecommand{\BIBforeignlanguage}[2]{{%
\expandafter\ifx\csname l@#1\endcsname\relax
\typeout{** WARNING: IEEEtran.bst: No hyphenation pattern has been}%
\typeout{** loaded for the language `#1'. Using the pattern for}%
\typeout{** the default language instead.}%
\else
\language=\csname l@#1\endcsname
\fi
#2}}
\providecommand{\BIBdecl}{\relax}
\BIBdecl

\bibitem{Huffman200300}
W.~C. Huffman and V.~Pless, \emph{Fundamentals of Error-Correcting
  Codes}.\hskip 1em plus 0.5em minus 0.4em\relax New York: Cambridge University
  Press, 2003.

\bibitem{Lin200400}
S.~Lin and D.~J. Costello, Jr., \emph{Error Control Coding}, 2nd~ed.\hskip 1em
  plus 0.5em minus 0.4em\relax Prentice Hall, 2004.

\bibitem{Gilbert195205}
E.~N. Gilbert, ``A comparison of signaling alphabets,'' \emph{Bell System
  Technical Journal}, vol.~31, no.~3, pp. 504--522, May 1952.

\bibitem{Varshamov195700}
R.~R. Varshamov, ``Estimate of the number of signals in error correcting
  codes,'' \emph{Doklady Akademii Nauk SSSR}, vol. 117, pp. 739--741, 1957.

\bibitem{Tsfasman198200}
M.~A. Tsfasman, S.~G. Vl\u{a}du\c{t}, and T.~Zink, ``Modular curves, {S}himura
  curves and {G}oppa codes, better than {V}arshamov-{G}ilbert bound,''
  \emph{Math. Nachr.}, vol. 109, pp. 21--28, 1982.

\bibitem{McEliece197703}
R.~J. McEliece, E.~R. Rodemich, H.~Rumsey, Jr., and L.~R. Welch, ``New upper
  bounds on the rate of a code via the {D}elsarte-{M}acwilliams inequalities,''
  \emph{{IEEE} Trans. Inf. Theory}, vol.~23, no.~2, pp. 157--166, Mar. 1977.

\bibitem{Jiang200408}
T.~Jiang and A.~Vardy, ``Asymptotic improvement of the {G}ilbert-{V}arshamov
  bound on the size of binary codes,'' \emph{{IEEE} Trans. Inf. Theory},
  vol.~50, no.~8, pp. 1655--1664, Aug. 2004.

\bibitem{Vu200509}
V.~Vu and L.~Wu, ``Improving the {G}ilbert-{V}arshamov bound for q-ary codes,''
  \emph{{IEEE} Trans. Inf. Theory}, vol.~51, no.~9, pp. 3200--3208, Sep. 2005.

\bibitem{Xing200101}
C.~Xing, ``Algebraic-geometry codes with asymptotic parameters better than the
  {G}ilbert-{V}arshamov and the {T}sfasman-{V}l\u{a}du\c{t}-{Z}ink bounds,''
  \emph{{IEEE} Trans. Inf. Theory}, vol.~47, no.~1, pp. 347--352, Jan. 2001.

\bibitem{Xing201109}
------, ``Asymptotically good nonlinear codes from algebraic curves,''
  \emph{{IEEE} Trans. Inf. Theory}, vol.~57, no.~9, pp. 5991--5995, Sep. 2011.

\bibitem{Cover199100}
T.~M. Cover and J.~A. Thomas, \emph{Elements of Information Theory}.\hskip 1em
  plus 0.5em minus 0.4em\relax John Wiley \& Sons, 1991.

\bibitem{Elias195503}
P.~Elias, ``Coding for noisy channels,'' \emph{IRE Conv. Rec.}, vol.~3, pp.
  37--46, Mar. 1955.

\bibitem{Gallager196300}
R.~G. Gallager, \emph{Low-Density Parity-Check Codes}.\hskip 1em plus 0.5em
  minus 0.4em\relax Cambridge, MA: MIT Press, 1963.

\bibitem{Barg200209}
A.~Barg and G.~D. Forney, Jr., ``Random codes: Minimum distances and error
  exponents,'' \emph{{IEEE} Trans. Inf. Theory}, vol.~48, no.~9, pp.
  2568--2573, Sep. 2002.

\bibitem{Yang200904}
S.~Yang, Y.~Chen, and P.~Qiu, ``Linear-codes-based lossless joint
  source-channel coding for multiple-access channels,'' \emph{{IEEE} Trans.
  Inf. Theory}, vol.~55, no.~4, pp. 1468--1486, Apr. 2009.

\bibitem{Arikan200907}
E.~Ar{\i}kan, ``Channel polarization: A method for constructing
  capacity-achieving codes for symmetric binary-input memoryless channels,''
  \emph{{IEEE} Trans. Inf. Theory}, vol.~55, no.~7, pp. 3051--3073, Jul. 2009.

\bibitem{Yang.M1}
S.~Yang, T.~Honold, Y.~Chen, Z.~Zhang, and P.~Qiu, ``Constructing linear
  encoders with good joint spectra,'' \emph{{IEEE} Trans. Inf. Theory},
  submitted for publication, draft available at http://arxiv.org/abs/0909.3131.

\bibitem{Yang201111}
------, ``Weight distributions of regular low-density parity-check codes over
  finite fields,'' \emph{{IEEE} Trans. Inf. Theory}, vol.~57, no.~11, pp.
  7507--7521, Nov. 2011.

\bibitem{Singleton196404}
R.~C. Singleton, ``Maximum distance q-nary codes,'' \emph{{IEEE} Trans. Inf.
  Theory}, vol.~10, no.~2, pp. 116--118, Apr. 1964.

\end{thebibliography}

\end{document}